\documentclass[a4paper,hidelinks, 12pt]{article}
\usepackage[utf8]{inputenc}
\usepackage{a4wide}
\usepackage{amsmath}
\usepackage{amsfonts}
\usepackage{amssymb}
\usepackage{amsthm}

\usepackage{subfigure}
\usepackage{cite}
\usepackage{xcolor}
\usepackage{soul}
\usepackage{cancel}
\usepackage[toc,page]{appendix}
\numberwithin{equation}{section}
\usepackage{mathtools}
\usepackage{hyperref}
\usepackage{enumitem}
\usepackage{algorithm}
\usepackage{float}

\allowdisplaybreaks

\newtheorem{theorem}{Theorem}
\newtheorem{prop}{Proposition}
\newtheorem{lemma}{Lemma}
\theoremstyle{plain}

\newcommand{\Su}{{\mathfrak{su}}}
\newcommand{\So}{{\mathfrak{so}}}

\begin{document}
\begin{titlepage}
\begin{center}

{\large \bf {Representation-theoretical characterization of canonical custodial symmetry in NHDM potentials}}

\vskip 1cm

R. Plantey\footnote{E-mail: Robin.Plantey@ntnu.no} 
 and
M. Aa. Solberg\footnote{E-mail: Marius.Solberg@ntnu.no} 

\vspace{1.0cm}

Department of Structural Engineering, \\ Norwegian University of Science and Technology, \\
Trondheim, Norway\\

\end{center}

\vskip 3cm

\begin{abstract}
By considering the basis-covariant constituents of $N$-Higgs-doublet potentials, we derive necessary and sufficient conditions for canonical $\mathsf{SO}(4)_C$ Custodial Symmetry (CS) of potentials with $N>2$ doublets, based on representation-theoretical and geometrical relations. In essence, our characterization relates the presence of canonical CS to basis-covariant vectors corresponding to particular bases of the defining representation of the orthogonal Lie algebras. For $N=3,4$ and $5$, the conditions demand little computational effort to be evaluated, and we provide practical algorithms that may be efficiently implemented in a computer program, for deciding whether or not a potential is 
custodial-symmetric.   
\end{abstract}

\end{titlepage}

\setcounter{footnote}{0}

\tableofcontents

\section{Introduction}
\label{sect:intro}

It is well known that extending the Standard Model (SM) with an arbitrary number of $\mathsf{SU}(2)_L$ doublets does not affect the value of the $\rho$ parameter
\begin{equation}
\rho = \frac{M_W^2}{M_Z^2\cos^2\theta_W} = 1
\end{equation}
at tree level, which is one of the reasons for the considerable attention that Multi-Higgs-Doublet Models (NHDMs) have received. 
The scalar potential of the SM has a related structural property, Custodial Symmetry (CS), which protects $\rho$ from large quantum corrections~\cite{Sikivie:1980hm}. CS is an accidental symmetry whereby the potential is invariant under the larger group $\mathsf{SO}(4)_C \simeq (\mathsf{SU}(2)_L\times \mathsf{SU}(2)_R)/\mathbb{Z}_2  \supset \mathsf{SU}(2)_L \times \mathsf{U}(1)_Y$. In the limit the hypercharge coupling $g'\to 0$ the kinetic terms are custodial-symmetric as well, and after spontaneous symmetry breaking $\mathsf{SO}(4)_C$ is broken down to custodial $\mathsf{SO}(3)_C$. Then the gauge bosons transform as a triplet under $\mathsf{SO}(3)_C$, and hence yields $m_W=m_Z$ and no electroweak mixing, to all orders of perturbation theory, when disregarding fermions. Due to the enhanced symmetry, approximate CS will keep $\rho$ near the experimentally measured magnitude, which is extremely close to one~\cite{Workman:2022ynf}.

Naturally, it is desirable to preserve these features in multi-Higgs doublet models. However, with more than one doublet, $\mathsf{SO}(4)_C$ is not an accidental symmetry of the potential anymore (and in addition, there are other possible symmetry breaking patterns, in contrast to the SM). Therefore one would like to identify the circumstances under which a NHDM potential is symmetric under $\mathsf{SO}(4)_C$. Nevertheless, this is a difficult task due to the basis freedom which can completely obfuscate a symmetry. In order to overcome basis freedom and identify $\mathsf{SO}(4)_C$ in a basis-independent way, we will characterize it using relations among basis-covariant objects, a powerful framework which has been successfully applied to other NHDM symmetries~\cite{Maniatis2008,Ivanov:2018ime,deMedeirosVarzielas:2019rrp,Plantey:2024yfm}.    

We will focus our attention on custodial transformations where $U_R\in \mathsf{SU}(2)_R$ acts as
\begin{equation}
\label{eq:canonical-CS}
\begin{pmatrix} i\sigma_2 \phi_i^* & \phi_i \end{pmatrix} \equiv B_{ii} \to B_{ii}U_R^\dag \, , \quad \forall i\in \{1,\ldots,N\}, 
\end{equation}
that is, has the same action on each bidoublet, in some doublet basis. There are, however, other inequivalent possibilities for CS~\cite{Battye:2011jj,Pilaftsis:2011ed,Darvishi:2019dbh}, e.g.~a 3HDM with $\mathsf{SU}(2)_R$ acting only on $B_{33}$, which are custodial in the sense that they may, with an appropriate symmetry breaking pattern, also protect $\rho$ from large corrections. However, the possible distinct $\mathsf{SU}(2)_R$ actions on the bidoublets will not be arbitrary~\cite{Gerard:2007kn}.  We do not explore these non-canonical possibilities here, and, unless otherwise specified, {\it from here on the term "custodial symmetry" will exclusively refer to canonical $\:\, \mathsf{SO}(4)_C$ custodial symmetry}, where the action of $\mathsf{SU}(2)_R$ is given by~\eqref{eq:canonical-CS} in some doublet basis. It was shown in~\cite{Nishi:2011gc} that for all custodial symmetries, where i) the Higgs kinetic term is left invariant, ii) $T_{3R}=\frac{1}{2}Y$ fixes $\mathsf{U}(1)_Y\subset \mathsf{SU}(2)_R$ and iii) $\mathsf{SU}(2)_R$ acts as $N$ copies of the defining representation, i.e.~as in \eqref{eq:canonical-CS} in some basis, the CS is equivalent to canonical $\mathsf{SO}(4)_C$, and the potential can be transformed into a characteristic form by a Higgs basis transformation. Thus, the problem of identifying canonical CS can be reduced to identifying this characteristic form of the potential. 
Different implementations of CS in the 2HDM have been introduced in~\cite{Pomarol:1993mu,Gerard:2007kn}, and were shown to be equivalent to canonical CS
in~\cite{Grzadkowski:2010dj,Haber:2010bw,Nishi:2011gc}. Different aspects of CS in models with
more than two doublets have been considered in~\cite{Olaussen:2010aq,Nishi:2011gc,Pilaftsis:2011ed,Solberg:2012au,Cen:2018wye,Solberg:2018aav}. With vacuum alignment in the direction of the CP-even fields, canonical CS in NHDMs will generate a mass degeneracy between charged and CP-odd sectors~\cite{Haber:2010bw,Olaussen:2010aq}. 
The present work is especially relevant for models with 3, 4 or 5 Higgs doublets.
In the 1970s, Weinberg presented a model with three doublets to accommodate spontaneous CP violation and natural flavour conservation~\cite{Weinberg:1976hu}. Since then, 
3HDMs have received significant attention. Models with four doublets have been considered in e.g.~\cite{Bjorken:1977vt,Kawase:2011az,Arroyo-Urena:2019lzv,Shao:2023oxt,PhysRevD.107.095001,Shao:2024ibu}, while 5HDMs 
in the context of higher-order CPs have been studied in~\cite{PhysRevD.98.015021}.  

While simple necessary and sufficient conditions for canonical CS can be formulated in the 2HDM in the bilinear formalism~\cite{Grzadkowski:2010dj}, the problem becomes more difficult with $N>2$ doublets. In this work we formulate general conditions for canonical $\mathsf{SO}(4)_C$ CS for a potential with any number of doublets. For $N=3$ doublets, our necessary and sufficient conditions are essentially the same as the known result where canonical CS is identified, in the adjoint space, by geometrical relations among the vectors which characterize the potential~\cite{Nishi:2006tg,Nishi:2011gc}. However, whereas these previous works used a combination of basis-invariants, we use covariant relations which, as we will see, generalize better and can be implemented in practical algorithms for testing whether a potential is custodial-symmetric. Indeed, we are able to devise practical procedures for detecting canonical CS in potentials with $N=4$ and $N=5$ doublets.

This paper is structured as follows. In Section~\ref{sect:method} we start by describing the covariants-based methods and proceed to derive a necessary and sufficient condition for canonical CS by making use of representation theory. Then, in Section~\ref{sect:conditions}, we show that our general condition can be implemented into practical algorithms for canonical CS detection in potentials with $N=3,4$ and $5$ doublets. Finally, our findings are summarized in Section~\ref{sec:Summary}. In Appendices~\ref{sec:SomeMathematicalResults} and~\ref{sec:EigenvalueDegeneraciesBeyondTheCharacteristic} we derive some auxiliary mathematical results and a method for handling the special case of potentials with large degeneracies, respectively.

\section{Method}
\label{sect:method}

This work relies on methods similar to those applied to order-2 $CP$ symmetry in~\cite{Plantey:2024yfm} where symmetries of the potential are characterized by representation-theoretical relations among basis-covariant objects. For completeness, and in order to set the notation, let us summarize this framework and recall important definitions.

We will write the potential for $N$ Higgs $\mathsf{SU}(2)_L$ doublets $\Phi_i$ in terms of gauge invariant bilinears 
\begin{equation}\label{E:bilinears}
K_0 = \Phi_i^\dag \Phi_i \;, \quad K_a = \Phi_i^\dagger (\lambda_a)_{ij} \Phi_j.
\end{equation}
where $K_a$, $a = 0, \ldots, N^2-1$ are given in terms of the generalized Gell-Mann matrices $\lambda_a$.
These matrices form a basis for the Lie algebra $\Su(N)$ and satisfy the commutation relations\footnote{In this basis the Killing form is proportional to the identity, hence we do not differentiate between upper and lower Lie algebra indices. Furthermore, we adopt the physicists' definition of a Lie algebra. For mathematicians a corresponding basis would be $\{i \lambda_j \}_{j=1}^{N^2-1}$.}
\begin{equation}
[\lambda_a, \lambda_b] = 2if_{abc}\lambda_c.
\end{equation}
For convenience, we order the generalized Gell-Mann matrices as in~\cite{Solberg:2018aav}, where the custodial-breaking bilinears appear first. That is
\begin{equation}
\label{eq:custorder}
\lambda_a^* = -\lambda_a \quad \text{for} \quad a=1,\ldots,k\equiv\frac{N(N-1)}{2}.
\end{equation}
Under a change of basis
\begin{equation}\label{E:HbasisCh}
\Phi_i \rightarrow U_{ij}\Phi_j\, , \quad  U\in \mathsf{SU}(N),
\end{equation}
it is readily seen that $K_0$ is a singlet while $K_a$ transforms under the adjoint representation
\begin{equation}
K_0 \rightarrow K_0 \,, \quad K_a\rightarrow R_{ab}(U)K_b
\end{equation}
with
\begin{equation}\label{E:R(U)}
R_{ab}(U) = \frac{1}{2}\text{Tr}(U^\dagger \lambda_a U \lambda_b).
\end{equation}
With these variables, the most general gauge invariant potential is then given by~\cite{Maniatis:2015gma}
\begin{equation}
\label{eq:Vbilinear}
V = M_0K_0 + M_aK_a + \Lambda_0 K_0^2 + L_a K_0K_a + \Lambda_{ab}K_aK_b
\end{equation}
and the coupling constants inherit from the bilinears simple transformation properties under a change of basis
\begin{align}
\label{E:LambdaTrafoU}
\Lambda &\to R(U) \Lambda R^T(U)\\
L &\to R(U) L \\
M &\to R(U) M 
\end{align}

Because basis transformations act on these couplings as the adjoint representation of $\mathsf{SU}(N)$, that is, the linear action of the group on its own Lie algebra, all the adjoint quantities which characterize the potential can be thought of as elements of $\Su(N)$. Thus, making use of this Lie algebra structure, it is possible to associate properties of the potential with representation theoretical relations inside of $\Su(N)$. Actually, as we will see in Section~\ref{sect:general}, a characteristic of CS is that a set of adjoint vectors forms a particular basis for the defining representation of $\So(N)$.

More formally, the mapping
\begin{align}
\Omega : \mathbb{R}^{N^2-1} &\rightarrow \Su(N)\\
\label{eq:isomorphism}
a &\mapsto a_i \lambda_i. 
\end{align}
defines an isomorphism between $\Su(N)$ and $\mathbb{R}^{N^2-1}$ when the latter is equipped with the product
\begin{align}
F : \mathbb{R}^{N^2-1}\times \mathbb{R}^{N^2-1} &\rightarrow \mathbb{R}^{N^2-1}\\
(a,b) &\mapsto f_{ijk}a_i b_j \equiv F_k^{(a,b)}
\end{align}
where $f_{ijk}$ are the structure constants of $\Su(N)$ in the Gell-Mann basis. Following the nomenclature of~\cite{deMedeirosVarzielas:2019rrp}, where it was used to identify 3HDM symmetries, we will refer to $F$ as the F-product. In what follows we will denote vectors of $\mathbb{R}^{N^2-1}$ with lower case letters and the associated $\Su(N)$ matrices by uppercase letters e.g.~$A \equiv a_i\lambda_i$. With these definitions, one has the following correspondence between commutators in $\Su(N)$ and F-products in $\mathbb{R}^{N^2-1}$
\begin{equation}
\label{eq:Fprod}
F^{(a,b)} = c \; \Leftrightarrow \; [A,B] = 2iC. 
\end{equation}
It is important to note that F-product relations are preserved by a change of Higgs basis $U$ i.e.
\begin{align}\label{E:FprodTransf}
 F^{(a,b)} = c \; \Leftrightarrow \; F^{(a',b')} = c',
\end{align}
where $x'=R(U)x$, cf.~\eqref{E:R(U)}, as is easily seen by considering the corresponding commutation relations.

\subsection{The custodial-symmetric potential}
\label{sec:cs-pot}

With the bilinears custodially ordered as in~(\ref{eq:custorder}), the potential is custodial-symmetric if and only if there exists a basis where $\Lambda$ assumes a block-diagonal form~\cite{Nishi:2011gc}
\begin{equation}
\label{eq:manifestCS}
\Lambda_C = 
\begin{pmatrix}
C_N & \mathbf{0}_{}\\
\mathbf{0} & A_N
\end{pmatrix}
\end{equation}
where $A_N$ is an arbitrary, real and symmetric $N^2-1-k\times N^2-1-k$ matrix and $C_N$ is a $k\times k$ matrix which we will refer to as the custodial block. For $N\leq 3$, the custodial block consists only of zeroes, corresponding to the absence of terms of the form 
\begin{equation}\label{eq:CC}
\widehat{C}_{ij}\widehat{C}_{kl} \equiv \text{Im}(\Phi_i^\dagger\Phi_j)\text{Im}(\Phi_k^\dagger\Phi_l)
\end{equation}
in $V$, however with more than three doublets additional custodial-invariant terms can be constructed~\cite{Nishi:2011gc} resulting in a non-zero custodial block (cf.~Section~\ref{sect:conditions} for explicit expressions of $C_N$). Note that CS imposes stronger constraints on the NHDM potential than order-2 CP symmetry which corresponds to the block structure~\eqref{eq:manifestCS} without any restrictions on the upper block~\cite{Nishi:2006tg,Nishi:2011gc,Plantey:2024yfm}.

The matrix $\Lambda$ in~(\ref{eq:Vbilinear}), being real and symmetric, can be written in terms of its eigenvalues and orthonormal set of eigenvectors, a form known as its spectral decomposition. This can always be done, even if the rotation that diagonalizes $\Lambda$ is not in $\text{Adj}(\mathsf{SU}(N)) \subset \mathsf{SO}(N^2-1)$. Let us therefore expand $\Lambda_C$ in terms of its eigenvalues and eigenvectors in a basis where the CS is manifest
\begin{equation}
\label{eq:spectraldec}
\Lambda_C = \sum_{a=1}^k \beta_a t_a t_a^T + \sum_{b=1}^{N^2-1-k} \gamma_b q_b q_b^T.
\end{equation}
with 
\begin{equation}
\label{eq:beta-t-def}
\sum_{a=1}^k \beta_a t_a t_a^T \equiv \begin{pmatrix}
C_N & \mathbf{0}_{}\\
\mathbf{0} & \mathbf{0}
\end{pmatrix},\quad
\sum_{b=1}^{N^2-1-k} \gamma_b q_b q_b^T \equiv \begin{pmatrix}
\mathbf{0} & \mathbf{0}_{}\\
\mathbf{0} & A_N
\end{pmatrix}.
\end{equation}
These important relations define the eigenvalues and eigenvectors, $\beta_a$ and $t_a$, which are used extensively in the remainder of the text. From~\eqref{eq:beta-t-def}, it can be seen that $\text{Span}(t_1,\ldots,t_k)=\text{Span}(e_1,\ldots,e_k)$ which, through the isomorphism~\eqref{eq:isomorphism}, corresponds to the (image of the) defining representation of $\So(N)$ within $\Su(N)$. 

On the other hand, the part of the potential~(\ref{eq:Vbilinear}) that is linear in the bilinears $K_a$ is determined by two adjoint vectors, $L$ and $M$. For a custodial-symmetric potential in a basis where the symmetry is manifest, the absence of custodial breaking terms implies
\begin{equation}
L\cdot t_a=M\cdot t_a=0, \quad \forall a\in {1,\ldots,k}. 
\end{equation}
We will use the same concise nomenclature as in~\cite{Plantey:2024yfm} and will refer to these conditions as $LM$-orthogonality.

Let us now take a closer look at the custodial block $C_N$ which, as we will see, determines for each $N$ the particular form of the conditions for CS. The bilinears $\widehat{C}$ from~\eqref{eq:CC} will in general break CS, but the combination
\begin{align}
 I^{(4)}_{abcd}=\widehat{C}_{ab}\widehat{C}_{cd}+\widehat{C}_{ad}\widehat{C}_{bc}+\widehat{C}_{ac}\widehat{C}_{db},
\end{align}
with $1\leq a,b,c,d \leq N$, will be invariant under CS, as shown by Nishi in~\cite{Nishi:2011gc}. $I^{(4)}$ is totally antisymmetric in all of its indices, and hence $I^{(4)}$ is zero if two indices are identical, so these terms will vanish in the 3HDM.
The most general, manifestly custodial-symmetric terms quadratic in the bilinears $\widehat{C}$ may then be written
\begin{align}
 V_{\widehat{C}^2}=\lambda_{abcd}I^{(4)}_{abcd},
\end{align}
with summation over repeated indices, and where we may (and will) take $a<b<c<d$ in the sum.

The custodial block $C_N$ will then be given by
\begin{align}
\label{eq:csblock-explicit}
 (C_N)_{ij}=\frac{1}{2} \frac{\partial^2  V_{\widehat{C}^2}}{\partial \widehat{C}_{m(i)n(i)} \partial \widehat{C}_{m(j)n(j)}},
\end{align}
 where $1 \leq i,j \leq k =N(N-1)/2$ and $(m(i),n(i))$ is a bijection between
the $k$ integers $i$ and the $k$ pairs $(m,n)$ with $1\leq m<n \leq N$. We will apply the
bijection which gives us the lexicographic order
\begin{align}
\{\widehat{C}_i\}_{i=1}^k=\{\widehat{C}_{12},\widehat{C}_{13},\ldots,\widehat{C}_{1 N},\widehat{C}_{23},\widehat{C}_{24},\ldots, \widehat{C}_{N-1,N}\},
\end{align} 
consistent with the order of the generalized Gell-Mann matrices referred to in Section~\ref{sect:method}, cf.~\cite{Solberg:2018aav}.

A careful inspection of~(\ref{eq:csblock-explicit}) reveals that the matrix structure of $C_N$ follows a fairly simple pattern when the number of doublets increases. For $N=4$, the smallest number of doublets with $V_{\widehat{C}^2} \neq 0$, the custodial block has an anti-diagonal structure
\begin{equation}
\label{eq:CSblock4-1}
C_4 = \lambda_{1234}\begin{pmatrix}
0&0&0&0&0&1 \\
0&0&0&0&-1&0 \\
0&0&0&1&0&0 \\
0&0&1&0&0&0 \\
0&-1&0&0&0&0 \\
1&0&0&0&0&0 \\
\end{pmatrix}
\end{equation}
while for $N>4$, the same anti-diagonal structure repeats once for all $\small \binom{N}{4}$ possible subsets of 4 distinct indices i.e.
\begin{equation}
\label{eq:CSblockN}
C_N = \sum_{a<b<c<d} \lambda_{abcd} D^{(abcd)}_N
\end{equation}
where $D^{(abcd)}_N$ is a $k\times k$ matrix which is zero everywhere except in the $6\times 6$ 
sub-block consisting of row and column numbers~$\big(i(a,b),i(a,c),i(a,d),i(b,c),i(b,d),i(c,d)\big)$, with $i(a,b)$ the lexicographic ordering bijection, where each sub-block has the anti-diagonal structure~(\ref{eq:CSblock4-1}).

\subsection{Representation and embedding indices}
\label{sect:indices}

Before deriving the representation-theoretical relations which characterize CS, let us recall some notions of Lie algebra theory related to the identification of representations. In $\Su(N)$ and $\So(N)$, one can define an inner product with
\begin{align}
\label{E:innerprod-su}
\langle X,Y \rangle_{\Su(N)} &\equiv \frac{1}{2}\text{Tr}(XY) = \frac{1}{4N}\text{Tr}(\text{ad}_X \text{ad}_Y), &&\forall X,Y\in \Su(N)\\
\label{E:innerprod-so}
\langle X,Y \rangle_{\So(N)} &\equiv \frac{1}{4}\text{Tr}(XY), &&\forall X,Y\in \So(N),\, N\geq4
\end{align}
where the numerical factors in front of the traces ensure a consistent normalization of the roots of both Lie algebras~\cite{Dynkin:1957um,Wang1985OnNH}. We have here chosen the inner products such that normalization is conserved by the mapping~\eqref{eq:isomorphism}, which infer long roots of the Lie algebra are normalized as well.  
  
Similarly, an inner product for a representation $\phi: \mathfrak{g} \rightarrow \mathfrak{gl}(n, \mathbb{C})$ can be defined by
\begin{align}
\langle \phi(X),\phi(Y) \rangle \equiv \frac{1}{2}\text{Tr}(\phi(X)\phi(Y)).
\end{align}

Having properly defined inner products, one can compute the so-called representation index of $\phi$
\begin{align}
\label{E:rep-index}
I_\phi \equiv \frac{\langle \phi(X),\phi(Y) \rangle}{\langle X,Y \rangle_\mathfrak{g}},
\end{align}
sometimes called Dynkin index, which is independent of $X,Y$ and can be used to characterize a representation~\cite{mckay1981tables}. From the definitions~(\ref{E:innerprod-su}),~(\ref{E:innerprod-so}) and~(\ref{E:rep-index}) it can be seen for example that the defining and adjoint representations of $\Su(N)$ have index $1$ and $2N$, respectively, while the defining representation of $\So(N)$, $\mathbf{N}$, has index\footnote{Recall that the defining representation of $\So(3)$ is equivalent to the adjoint representation of $\Su(2)$.}
\begin{align}
\label{E:index-so}
I_\mathbf{N} \equiv
\begin{cases}
4\,, & N=3\\
2\,, & N>3
\end{cases}
.
\end{align}
In what follows, we will consider subalgebras of $\Su(N)$ and their embeddings into the defining representation of $\Su(N)$. An embedding of a subalgebra $\mathfrak{h}$ is a faithful Lie algebra homo\-morphism~$p: \mathfrak{h} \rightarrow \Su(N)$ and inequivalent embeddings are characterized by the so-called embedding index
\begin{align}
\label{E:embindex}
J_p = \frac{\langle p(X),p(Y)\rangle_{\Su(N)}}{\langle X,Y \rangle_{\mathfrak{h}}}\,, \quad\quad \forall X,Y\in \mathfrak{h}.
\end{align}
Given such a subalgebra embedding and a representation $\phi$ of $\Su(N)$, the composition $\phi\circ p$ furnishes a representation of $\mathfrak{h}$. The representation and embedding indices are related by~\cite{Dynkin:1957um}
\begin{align}
J_p = \frac{I_{\phi p}}{I_\phi}.
\end{align}
In particular, if $\phi$ is the defining representation of $\Su(N)$ then $J_p = I_{\phi p}$. If $\phi\circ p$ is reducible, the index will be the sum of the indices of the irreducible components.

As an example, let us consider the embeddings of $\So(N)$ into $\Su(N)$ which are of special interest in this work. Consider a normalized basis of $\So(N)$, $\{X_a\}_{a=1}^k$, satisfying commutation relations
\begin{align}
[X_a,X_b] \equiv 2ig_{abc} X_c.
\end{align}
An embedding $p$ into $\Su(N)$ naturally preserves these commutation relations, but it may not preserve the normalization of the basis elements. Indeed, according to~(\ref{E:embindex}), the embedded basis elements $\{p(X_a)\}_{a=1}^k$ have norm $\sqrt{J_p}$ in $\Su(N)$. Hence the normalized embedded basis $\{p(\bar X_a) \equiv p\big(\frac{X_a}{\sqrt{J_p}}\big)\}_{a=1}^k$ satisfies the commutation relations
\begin{align}
\label{E:norm-comm}
\sqrt{J_p}\, [p(\bar X_a), p(\bar X_b)] = 2ig_{abc} p(\bar X_c).
\end{align}
The point is that if one has found a subalgebra, e.g.~$\So(N)$ in $\Su(N)$, then information about the embedding and representation can be extracted by consistently normalizing a basis of the subalgebra since this makes the embedding index apparent. Particularly relevant to this work is the embedding of the defining representation of $\So(N)$ into $\Su(N)$ furnished by the antisymmetric Gell-Mann matrices $\{\lambda_a\}_{a=1}^k$. In that case the index of the relevant embedding, $J_p$, equals the index of the defining representation of $\So(N)$, $I_N$, given in~(\ref{E:index-so}) and so~(\ref{E:norm-comm}) yields the following relation between the structure constants of $\So(N)$ and $\Su(N)$ in the Gell-Mann basis, $g_{abc}$ and $f_{abc}$
\begin{align}
g_{abc} = \sqrt{I_N}f_{abc}\, , \quad a,b,c=1,\ldots,k.
\end{align}

\subsection{A general necessary and sufficient condition for custodial symmetry}
\label{sect:general}

From the spectral decomposition~(\ref{eq:spectraldec}) we can deduce a basis-invariant signature of CS, namely, in the presence of CS, $\Lambda$ has $k$ $LM$-orthogonal eigenvectors $t_a$ with special eigenvalues $\beta_a$, spanning the subspace $\text{Span}(e_1,\ldots,e_k)$ in some Higgs basis. The eigenvalues $\beta_a$ and the eigenvectors' components $(t_a)_b$ depend on the number of doublets and can be calculated by considering the most general custodial-symmetric NHDM in a basis where the symmetry is manifest (cf.~Section~\ref{sect:conditions}).

Now $\text{Span}(e_1,\ldots,e_k)$ is isomorphic to the defining representation of $\So(N)$, which means that the matrices 
\begin{equation}\label{E:tTcorr}
T_a = (t_a)_b \lambda_b
\end{equation}
form a basis for the defining representation of $\So(N)$. Depending on the components of $t_a$, their commutation relations can be different from the usual $\So(N)$ commutation relations $\sqrt{I_N}[\lambda_a, \lambda_b] = 2ig_{abc}\lambda_c$ and in general we will have 
\begin{equation}
\label{eq:comm1}
\sqrt{I_N}[T_a,T_b] = 2ig'_{abc}T_c.
\end{equation}
where $I_N$ is the representation index of the defining representation of $\So(N)$, $g'_{abc} = \sqrt{I_N}t_{ad}t_{be}t_{cf}f_{def}$ and we abbreviate the components $(t_a)_b \equiv t_{ab}$ from now on. Equivalently, in $\mathbb{R}^{N^2-1}$ we have, according to~(\ref{eq:Fprod}), the F-product relations
\begin{equation}
\label{eq:Fprod1}
\sqrt{I_N} F^{(t_a, t_b)} = g'_{abc}t_c.
\end{equation}
This property, being a vector relation, can be verified in any Higgs basis, cf.~\eqref{E:FprodTransf}. Indeed, under a basis change $U \in \mathsf{SU}(N)$
\begin{equation}
T_a \rightarrow UT_aU^\dagger \equiv V_a \; \Leftrightarrow \; t_{ab} \rightarrow R(U)_{bc}t_{ac} \equiv v_{ab},
\end{equation}
and the commutation and F-product relations take the same form as~(\ref{eq:comm1}) and~(\ref{eq:Fprod1}).

We note that this characterization of CS is the same as that of $CP2$ symmetry derived in~\cite{Plantey:2024yfm} strengthened with restrictions on the eigenvalues and the F-product relations of the $LM$-orthogonal eigenvectors which span the defining representation of $\So(N)$. While it would be possible to detect CS by first establishing CP2 and then checking if the eigenvalues and F-product relations are consistent with CS, we will now show that a much simpler procedure, based on embedding indices, can be devised.

To prove that the conditions given above are also sufficient we will make use of a result proved in~\cite{Plantey:2024yfm}, namely, inside $\Su(N)$, there are no $\So(N)$ subalgebras apart from the defining representation\footnote{Both representations and subalgebras of Lie algebras are defined as Lie algebra
homomorphisms. The only distinction between the two concepts is that a Lie subalgebra
always corresponds to an injective (one-to-one) homomorphism whose image lies within the ambient algebra, whereas representations do not in general respect this restriction.}, except for $N=3,4,5,6,8$ for which the alternative $\So(N)$ subalgebras are listed in Table~\ref{tab:alternative-indices}. We also include the embedding indices of these subalgebras in $\Su(N)$, calculated with \texttt{LieART}~\cite{Feger:2019tvk}.

\begin{table}[H]
\centering
\begin{tabular}{|l|l|l|}
\hline
Dimension & Representation & Index\\
\hline
$N=3$ & $\mathbf{2}+\mathbf{1}$ & 1 \\
\hline
$N=4$ & $\mathbf{2}+\mathbf{2'} $ & 1 \\
\hline
$N=5$ & $\mathbf{4}+\mathbf{1}$ & 1 \\
\hline
$N=6$ & $\mathbf{4}+\mathbf{1}+\mathbf{1}$ & 1 \\
    & $\mathbf{\bar{4}}+\mathbf{1}+\mathbf{1}$ & 1 \\
\hline
$N=8$ & $\mathbf{{8}_s}$ & 2\\
    & $\mathbf{{8}_c}$ & 2\\
\hline
\end{tabular}
\caption{$\So(N)$ subalgebras different from the defining representation and their embedding indices in $\Su(N)$. For reference, the defining representation has index $2$ for $N>3$ and index $4$ for $N=3$, cf.~\eqref{E:index-so}.}
\label{tab:alternative-indices}
\end{table}

We can now state and prove the sufficiency of our general condition which links CS to special bases of the defining representation of $\mathfrak{so}(N)$.
\begin{theorem}
\label{thm:cs}
Let $N>2$. Then an NHDM potential is custodial-symmetric if and only if the matrix $\Lambda$ has $k=N(N-1)/2 \;$ $LM$-orthogonal normalized eigenvectors $v_a$, with the same eigenvalues and F-product relations as the normalized eigenvectors $t_a$ of some instance of the custodial block $C_N$.
\begin{proof} 
$\pmb{(\Leftarrow):}$  The linear mapping $p$ given by $p(T_a)=V_a$ is a Lie algebra homomorphism of $\So(N)$ into $\Su(N)$: It respects the commutator, in the sense $p([T_a,T_b])=[p(T_a),p(T_b)]$, since, by assumption, also the normalized eigenvectors $v_a$ satisfy the F-product relations~\eqref{eq:Fprod1},
i.e.
\begin{align}
	\sqrt{I_N} F^{(v_a, v_b)} = g'_{abc}v_c,
\end{align}
where $g'_{abc}$ are known numbers.\footnote{Here $p$ may be extended beyond $\So(N)$ by linearity. Hence $p([T_a,T_b])=-ip(i[T_a,T_b])$, since $[T_a,T_b]$ strictly speaking is not an element of $\So(N)$, if we insist on applying the physicists' definition of a Lie algebra.}
 Moreover, it is faithful due to Proposition~\ref{prop:faithful} in Appendix \ref{sec:SomeMathematicalResults}, so $p$ is a subalgebra embedding of $\So(N)$. This subalgebra will correspond to the defining representation, since the embedding index $I_N$ is unique for the defining representation.
Indeed, if $N=3,4,5,6$ then the other possible $\So(N)$ subalgebras have embedding index (cf.~Table~\ref{tab:alternative-indices}) different from that of the defining representation~\eqref{E:index-so}. In the very special case $N=8$ $\So(8)$ has three inequivalent representations with the same embedding index\footnote{This is a consequence of triality, a peculiar feature only present in $\So(8)$ which originates in the exceptionally large symmetry of the $D_4$ root system~\cite{mckay1981tables,zee2016group}.} in $\Su(8)$. However, the images of the representations $\mathbf{{8}_s}$ and $\mathbf{{8}_c}$ are the same
as for the defining representation $\mathbf{{8}_v}$, so a detection of any of these representations will correspond to the matrices of $\mathbf{{8}_v}$.\footnote{I.e.~which of these representations of $\So(8)$ that are detected in our context will just be a matter of convention. The authors are grateful to Andreas Trautner for pointing out this fact.}  

In all other cases, the defining representation is the only $\So(N)$ subalgebra according to Proposition~2 in~\cite{Plantey:2024yfm} quoted in Table~\ref{tab:alternative-indices} above. The representation generated by $\{V_a\}$ is therefore equivalent to the $\{T_a\}$ representation and moreover, using Proposition~\ref{prop:unitary-equivRedIrred} from the Appendix, the equivalence is provided by a unitary matrix $U$ and thus can be achieved by a change of Higgs basis. 
Hence $T_a=U V_a U^\dag$ and $t_a=R(U)v_a$, and writing the rotated $\Lambda$ using its spectral decomposition, we get $\Lambda' = R\Lambda R^T=\beta_a  Rv_a v_a^T R^T = \beta_a t_a t_a^T $ where $a$ is summed up to $N^2-1$ and only the first $k$ eigenvectors are relevant for the custodial structure. Finally, since $\{t_a\}_{a=1}^k$ and $\{\beta_a\}_{a=1}^k$ corresponded to an instance of the custodial block $C_N$, 
$\Lambda'$ is manifestly custodial-symmetric.  
 Hence there are no custodial-breaking terms quadratic in the bilinears $K_a$. Moreover, since $v_a\cdot L = v_a\cdot M = 0$ for $a\leq k$, we have in the primed basis $L,M \in \text{Span}(e_{k+1},\ldots,e_{N^2-1})$ and there are no custodial-breaking terms linear in $K_a$ either. Therefore the potential is custodial-symmetric.  

$\pmb{(\Rightarrow):}$ This follows from the arguments given at the beginning of this section.
\end{proof}
\end{theorem}

\section{Conditions for custodial symmetry}
\label{sect:conditions}

In this section we show how CS can be detected in practice, starting with the known case of the 3HDM~\cite{Nishi:2011gc} and then moving on to the 4HDM and 5HDM. For these models, the eigenvectors of the custodial block (which is described in Section~\ref{sec:cs-pot}) take a simple form and all the Lie algebra bases of the defining $\So(N)$ representation which correspond to CS can be identified, allowing for a practical application of Theorem~\ref{thm:cs}. The concrete algorithms which we introduce below can be implemented numerically to decide if a parameter space point of a potential is custodial-symmetric, although analytical implementations are possible for sufficiently simple potentials. In the latter case, the existence of CS can be established at once for all possible values of the parameters.

\subsection{$N=3$}
\label{sec:N3}

For the 3HDM the custodial block consists only of zeroes
\begin{equation}
C_3 = 
\begin{pmatrix}
0&0&0\\
0&0&0\\
0&0&0
\end{pmatrix}
\end{equation}
and hence the eigenvectors and eigenvalues for the custodial block in the spectral decomposition of $\Lambda_C$~(\ref{eq:spectraldec}) are simply given by
\begin{equation}
t_{ai} = \delta_{ai}, \quad \beta_a = 0,\quad a=1,2,3.
\end{equation}
These normalized vectors satisfy the F-product relations
\begin{equation}
\label{eq:Fprod3}
2F^{(t_a,t_b)} = \epsilon_{abc} t_c
\end{equation}
since the associated matrices $T_a$ are simply given by the Gell-Mann matrices $\lambda_1, \lambda_2, \lambda_3$ and obey the commutation relations $[T_a,T_b] = i\epsilon_{abc} T_c$. According to~\eqref{eq:Fprod1}, one can read off the index $\sqrt{I_3}=2$ in~\eqref{eq:Fprod3} which signals an embedding of the defining representation of $\So(3)$, $\mathbf{3}$, in $\Su(3)$ (cf.~\eqref{E:index-so}). We note, in passing, that $\mathbf{3}$ has previously been distinguished from $\mathbf{2+1}$ using a generalized pseudoscalar~\cite{Nishi:2011gc}
\begin{equation}
\label{eq:nishi-pseudoscalar}
I(t_a,t_b,t_c) \equiv F(t_a,t_b)\cdot t_c
\end{equation}
with the values $\frac{1}{2}$ and $1$ characterizing $\mathbf{3}$ and $\mathbf{2+1}$. These numerical values are determined by embedding indices and, in particular, it is easy to see that~\eqref{eq:nishi-pseudoscalar} follows from~\eqref{eq:Fprod3}.

Applying Theorem~\ref{thm:cs}, we can now devise a practical procedure for verifying whether a 3HDM is custodial-symmetric which is summarized in Algorithm~\ref{alg:CS3HDM}. 
\begin{algorithm*}[!htbp]
\caption{Determining if a 3HDM potential has a CS}
\label{alg:CS3HDM}

\begin{enumerate}[align=left, label=\fbox{\arabic*}]
\item If dim(ker($\Lambda))\geq 3$ proceed, else return False.
\item Let $W_0^{LM}$ be the $LM$-orthogonal subspace of $\text{ker}(\Lambda)$. If $\text{dim}(W_0^{LM})\geq 3$ proceed, else return False.
\item If there exists three orthonormal vectors in $W_0^{LM}$ satisfying the F-product relations~\eqref{eq:Fprod3} return True, else return False.
\end{enumerate}
\end{algorithm*}

A nice feature of Algorithm~\ref{alg:CS3HDM} is that, if $\text{dim}(W_0^{LM})=3$, then the F-product relations~\eqref{eq:Fprod3} are satisfied in any orthonormal basis of $W_0^{LM}$ due to the invariance of these relations under rotation, cf.~Proposition~\ref{prop:rotation-invariance} in the Appendix. In the event that $\text{dim}(W_0^{LM})>3$, identifying a set of $LM$-orthonormal nullvectors satisfying the right F-products may be non-trivial. In Appendix~\ref{sec:EigenvalueDegeneraciesBeyondTheCharacteristic} we illustrate how this may be done by numerically solving a set of F-product closure equations for three orthonormal vectors in $W_0^{LM}$.

\subsection{$N=4$}
\label{sec:N4}

With four doublets the custodial block now takes the form 
\begin{equation}
\label{eq:CSblock4}
C_4 = \alpha\begin{pmatrix}
0&0&0&0&0&1 \\
0&0&0&0&-1&0 \\
0&0&0&1&0&0 \\
0&0&1&0&0&0 \\
0&-1&0&0&0&0 \\
1&0&0&0&0&0 \\
\end{pmatrix}
\end{equation}
with $\alpha$ a real constant. The cases $\alpha \neq 0$ and $\alpha=0$ have different basis-invariant signatures of CS and hence are identified by different conditions. To determine whether an arbitrary potential has CS, both manifestations, $\alpha \neq 0$ and $\alpha=0$, must be checked as described below. 

\subsubsection*{The case $\alpha\neq 0$}

With $\alpha\neq 0$, CS implies the matrix $\Lambda$ has two sets of threefold degenerate eigenvectors with eigenvalues $\pm \alpha$. In the basis were the symmetry is manifest, these eigenvectors have components
\begin{align}
t_1^+ &= \frac{1}{\sqrt{2}}(+1,0,0,0,0,-1,\mathbf{0}_9)^T \nonumber \\
t_2^+ &= \frac{1}{\sqrt{2}}(0,+1,0,0,+1,0,\mathbf{0}_9)^T  \nonumber \\
t_3^+ &= \frac{1}{\sqrt{2}}(0,0,-1,+1,0,0,\mathbf{0}_9)^T\\
t_1^- &= \frac{1}{\sqrt{2}}(+1,0,0,0,0,+1,\mathbf{0}_9)^T \nonumber \\
t_2^- &= \frac{1}{\sqrt{2}}(0,+1,0,0,-1,0,\mathbf{0}_9)^T \nonumber \\
t_3^- &= \frac{1}{\sqrt{2}}(0,0,+1,+1,0,0,\mathbf{0}_9)^T \nonumber 
\end{align}
and one finds that they satisfy the $\So(3)\oplus \So(3)$ F-product relations 
\begin{align}
\label{eq:Fprod4}
\sqrt{2}F^{(t_a^\pm, t_b^\pm)} &= \epsilon_{abc}t_c^\pm \\
F^{(t_a^\pm, t_b^\mp)} &= 0 \nonumber
\end{align}
which, as expected, come with index $\sqrt{I_\mathbf{4}}=\sqrt{2}$ and correspond to an embedding of the defining representation of $\So(4)$. The only other faithful representation which could arise, $\mathbf{2 + 2'}$~\cite{Plantey:2024yfm}, would have index $I_\mathbf{2 + 2'}=1$ 
(cf.~Table~\ref{tab:alternative-indices} or \cite{mckay1981tables}) and could be easily discarded. 

Having found the F-product relations characterizing CS, Theorem~\ref{thm:cs} can be implemented into Algorithm~\ref{alg:CS4HDM-nondegen} to detect $\alpha\neq 0$ instances of CS. Note that there may be several pairs of threefold degenerate eigenvalues $\pm \alpha$, and the algorithm must be applied once for each possible $\alpha$. 

\begin{algorithm*}[!htbp]
\caption{Determining if a 4HDM potential has a CS ($\alpha \neq 0$)}
\label{alg:CS4HDM-nondegen}

\begin{enumerate}[align=left, label=\fbox{\arabic*}]
\item If, for any $\alpha \in \mathbb{R}$, $\Lambda$ has two eigenvalues $-\alpha$ and $\alpha$ such that $\text{dim}(W_{\pm \alpha}) \geq 3$ for both eigenvalue spaces $W_{\pm \alpha}$, proceed. Else return False. 
\item Let $W_{\pm \alpha}^{LM}$ be the $LM$-orthogonal subspaces of $W_{\pm \alpha}$. If $\text{dim}(W_{\pm \alpha}^{LM})\geq 3$ proceed, else return False.
\item If two subsets of three basis vectors $v_a^{\pm}$ satisfy the F-product relations~\eqref{eq:Fprod4} return True, else return False.
\end{enumerate}
\end{algorithm*}

The F-product relations~(\ref{eq:Fprod4}) determine all the Lie algebra bases which correspond to CS for this model and it is remarkable that, when $\text{dim}(W_{+\alpha}^{LM})=\text{dim}(W_{-\alpha}^{LM})=3$, analogously to the case of the 3HDM, these relations are independent of the choice of bases for the $LM$-orthogonal degenerate subspaces $W_{\pm \alpha}^{LM}$, as follows directly from Proposition~\ref{prop:rotation-invariance}. When there are extra degeneracies and $\text{dim}(W_{\pm\alpha}^{LM}) > 3$, the techniques of Appendix~\ref{sec:EigenvalueDegeneraciesBeyondTheCharacteristic} may be necessary in step 3 of Algorithm~\ref{alg:CS4HDM-nondegen} to isolate two sets of orthonormal eigenvectors satisfying the F-products~\eqref{eq:Fprod4}.

\subsubsection*{The case $\alpha=0$}
\label{sec:degen4}

It may happen that the potential under consideration corresponds to an instance of CS where $\alpha=0$. In that case $\Lambda$ has 6 nullvectors generating the defining representation of $\So(4)$. Note that, in contrast with the case $\alpha \neq 0$, any basis of the defining representation of $\So(4)$ will correspond to CS. Therefore there are no particular F-product relations to be checked. 
Instead one must verify whether or not the 6 nullvectors induce the defining representation of $\So(4)$. This is a slightly stronger condition than the existence of an order-2 $CP$ symmetry~\cite{Plantey:2024yfm}, thus $\alpha=0$ manifestations of CS can be checked by applying Algorithm~3 from~\cite{Plantey:2024yfm} and restricting the candidate eigenvectors to nullvectors. In case of more than 6 nullvectors, the methods of Appendix~\ref{sec:EigenvalueDegeneraciesBeyondTheCharacteristic} may be applied.

In an earlier work \cite{Solberg:2018aav} on CS by one of the authors, the custodially invariant terms $I^4_{abcd}$ were not included, and hence, conditions only for the cases of the type $\alpha= 0$ (which automatically holds for the 3HDM) were given. Thus, the present work supersedes \cite{Solberg:2018aav}. Moreover, the conditions in the present article are far more analytical than the conditions in \cite{Solberg:2018aav} since they, in the absence of extended degeneracies, do not rely on solving large systems of quadratic equations. Therefore, the methods of the present article may be more efficient in several cases, in addition to being complete and less numerical in nature. Nevertheless, in the presence of extended degeneracies, like 7 nullvectors for $N=4$, the numerical methods of Appendix~\ref{sec:EigenvalueDegeneraciesBeyondTheCharacteristic} together with the conditions of \cite{Solberg:2018aav}, might just as efficiently determine if a potential is custodial-symmetric, since we in this case will have to find the minimum of a quartic polynomial (the cost function) in both approaches, cf.~Appendix~\ref{sec:EigenvalueDegeneraciesBeyondTheCharacteristic}. Anyway, \cite{Solberg:2018aav} will here yield the same results as the present article. However, applying the original numerical methods of \cite{Solberg:2018aav} will be significantly more computationally demanding.       

\subsection{$N=5$}
\label{sec:N5}
Increasing the number of doublets to five, the number of free parameters $\lambda_{abcd}$ in the custodial block increases to $\small \binom{5}{4}= 5$ which seems to make the detection of CS more difficult as the eigenvectors of $C_5$ are not constants as they were for $N=4$ and $N=3$. However, we show in Proposition~\ref{prop:5HDM-rotations} in the Appendix that $C_5$ always can be transformed into the form
\begin{gather}
\label{eq:CSblock5}
C_5 = \alpha D_5^{1234} = \alpha
\begin{pmatrix}
0&0&0&0&0&0&0&1&0&0\\
0&0&0&0&0&-1&0&0&0&0\\
0&0&0&0&1&0&0&0&0&0\\
0&0&0&0&0&0&0&0&0&0\\
0&0&1&0&0&0&0&0&0&0\\
0&-1&0&0&0&0&0&0&0&0\\
0&0&0&0&0&0&0&0&0&0\\
1&0&0&0&0&0&0&0&0&0\\
0&0&0&0&0&0&0&0&0&0\\
0&0&0&0&0&0&0&0&0&0\\
\end{pmatrix}
\end{gather}
by a rotation of the doublets. Therefore all instances of CS for the 5HDM are equivalent to~(\ref{eq:CSblock5}). As in the 4HDM, the cases $\alpha=0$ and $\alpha \neq 0$ must be treated separately. Moreover, we note that $C_5$ in~\eqref{eq:CSblock5} is identical to $C_4$ in~\eqref{eq:CSblock4} if the four zero rows and columns are removed. Hence we get the same eigenvalue pattern as for $N=4$, except for 4 additional nullvectors. These characteristic eigenvalue degeneracies are also mentioned in reference~\cite{Nishi:2011gc}.

\subsubsection*{The case $\alpha \neq 0$}

Thanks to the equivalences among all instances of CS discussed above, the constant eigenvectors of~(\ref{eq:CSblock5}) 
\begin{align}
\label{eq:evecs5}
t_1^+ &= \frac{1}{\sqrt{2}}(+1, 0, 0, 0, 0, 0, 0, -1, 0, 0, \mathbf{0}_{14})^T \nonumber \\
t_2^+ &= \frac{1}{\sqrt{2}}(0, +1, 0, 0, 0, +1, 0, 0, 0, 0, \mathbf{0}_{14})^T \nonumber \\
t_3^+ &= \frac{1}{\sqrt{2}}(0, 0, -1, 0, +1, 0, 0, 0, 0, 0, \mathbf{0}_{14})^T \nonumber \\
t_1^- &= \frac{1}{\sqrt{2}}(+1, 0, 0, 0, 0, 0, 0, +1, 0, 0, \mathbf{0}_{14})^T \nonumber \\
t_2^- &= \frac{1}{\sqrt{2}}(0, +1, 0, 0, 0, -1, 0, 0, 0, 0, \mathbf{0}_{14})^T \nonumber \\
t_3^- &= \frac{1}{\sqrt{2}}(0, 0, +1, 0, +1, 0, 0, 0, 0, 0, \mathbf{0}_{14})^T 
\end{align}
including nullvectors
\begin{align}
\label{eq:evecs5null}
n_1 &= (0, 0, 0, +1, 0, 0, 0, 0, 0, 0, \mathbf{0}_{14})^T \nonumber \\
n_2 &= (0, 0, 0, 0, 0, 0, +1, 0, 0, 0, \mathbf{0}_{14})^T \nonumber \\
n_3 &= (0, 0, 0, 0, 0, 0, 0, 0, +1, 0, \mathbf{0}_{14})^T \nonumber \\
n_4 &= (0, 0, 0, 0, 0, 0, 0, 0, 0, +1, \mathbf{0}_{14})^T 
\end{align}
characterize CS in the 5HDM. In~(\ref{eq:evecs5}) the eigenvectors $t_a^\pm,\,(a=1,2,3)$ have eigenvalue $\pm \alpha$ and satisfy $\So(4)\cong \So(3)\oplus \So(3)$ F-products
\begin{align}
\label{eq:Fprod5}
\sqrt{2} F^{(t_a^\pm, t_b^\pm)} &= \epsilon_{abc}t_c^\pm \\
F^{(t_a^\pm, t_b^\mp)} &= 0 \nonumber.
\end{align}
The F-product relations involving the nullvectors $n_a$ are not meaningful in practice since they depend on which basis is chosen for the nullspace. Without all the F-products one cannot establish whether or not a given set of $10$ eigenvectors spans $\So(5)$. Indeed, even if~(\ref{eq:Fprod5}) is satisfied, it may be that the 10 eigenvectors do not generate a subalgebra i.e.~do not close under the F-product. To ensure that one has found a 10-dimensional subalgebra one can use projectors as follows. 
Let $v_a$ be a set of 10 candidate orthonormal eigenvectors of $\Lambda$, this set closes under the F-product if and only if
\begin{align}
\label{eq:FprodClosure}
(I - P_0) F^{(v_a,v_b)} = 0\, , \quad \forall a,b\in \{1,\ldots,10\}
\end{align}
where $I$ is the identity matrix and $P_0 = \sum_{a=1}^{10} v_a v_a^T$ is a projector onto the subspace spanned by this subset of eigenvectors.

Analyzing the subalgebras of the classical Lie algebras~\cite{Dynkin:1957um, Feger:2019tvk, Lorente:1972xw} one finds that the only 10d $\Su(5)$ subalgebra containing an $\So(4)$ subalgebra is $\So(5) \cong \mathfrak{sp}(4)$. Moreover, the prefactor $\sqrt{I_{\mathbf{5}}}=\sqrt{2}$ in~\eqref{eq:Fprod5} corresponds to the embedding index of the defining representation of $\So(5)$ in $\Su(5)$. Thus, these F-product relations, although incomplete, are still sufficient to establish CS using Theorem~\ref{thm:cs}, provided that the eigenvalue pattern is correct and that the six eigenvectors completed with 4 nullvectors form a subalgebra. The practical steps for checking $\alpha \neq 0$  instances of CS of a 5HDM potential are given in Algorithm~\ref{alg:CS5HDM-nondegen} below. As in the $N=4$ case, this algorithm is to be applied once for each pair of threefold degenerate eigenvalues $\alpha$.
\begin{algorithm*}[!htbp]
\caption{Determining if a 5HDM potential has a CS ($\alpha \neq 0$)}
\label{alg:CS5HDM-nondegen}

\begin{enumerate}[align=left, label=\fbox{\arabic*}]
\item If, for any $\alpha \in \mathbb{R}$, $\Lambda$ has two eigenvalues $-\alpha$ and $\alpha$ such that $\text{dim}(W_{\pm \alpha}) \geq 3$ for both eigenvalue spaces $W_{\pm \alpha}$, and $\text{dim}(W_0) \geq 4$, proceed. Else return False. 
\item Let $W^{LM}$ be the $LM$-orthogonal subspaces of $W$. If $\text{dim}(W_{\pm \alpha}^{LM})\geq 3$ and $\text{dim}(W_0^{LM})\geq 4$ proceed, else return False.
\item If two subsets of three orthonormal vectors of $W_{\pm\alpha}^{LM}$ satisfy the F-product relations~\eqref{eq:Fprod5} and can be completed by four vectors of $W^{LM}_0$ into a 10d subalgebra, return True. Else return False
\end{enumerate}
\end{algorithm*}

It should be noted that, as before, this CS test relies on verifying whether sets of three degenerate eigenvectors satisfy $\So(3)$ F-product relations, which are independent of the choice of orthonormal basis for the degenerate subspace. It is again this fact that makes the test practical to implement. Moreover, if $\text{dim}(W^{LM}_0) = 4$ then, in step 4 of Algorithm~\ref{alg:CS5HDM-nondegen}, closure can be checked in any basis of $W^{LM}_0$. However, cases with extended degeneracies such that either $\text{dim}(W^{LM}_0) > 4$ or $\text{dim}(W^{LM}_{\pm\alpha}) > 3$ require special treatment, which is described in Appendix~\ref{sec:EigenvalueDegeneraciesBeyondTheCharacteristic}.

\subsubsection*{The case $\alpha = 0$}
\label{sec:degen5}
To check whether the potential corresponds to an instance of CS with $\alpha=0$ one has to check whether a set of 10 nullvectors gives the defining representation of $\So(5)$. This can be done in exactly the same way as for the 4HDM (cf.~\ref{sec:degen4}).   

\subsection{$N > 5$}
\label{sec:beyondN5}

As the number of doublets increases, the basis-invariant signatures of CS become more and more subtle. This is because the number of parameters in the custodial block grows as $\small 
\binom{N}{4}$ and its eigenvalues and eigenvectors become functions of more and more parameters, removing the possibility for clear patterns. Hence it becomes increasingly difficult to detect CS. An exception is instances of CS where all the eigenvalues of the custodial block are zero. Then CS can be identified, exactly as with $N=4$ and $N=5$ (cf.~sections~\ref{sec:degen4}), by applying the CP2 detection methods of~\cite{Plantey:2024yfm} restricted to nullvectors of $\Lambda$. In the case of more than $k=N(N-1)/2$ nullvectors, techniques like the ones found in Appendix~\ref{sec:EigenvalueDegeneraciesBeyondTheCharacteristic} may be invoked. For the remaining instances of CS, where some eigenvalues of the custodial block are non-zero, corresponding to the presence of terms~\eqref{eq:CC} in the potential, we outline below the difficulties that arise beyond $N=5$ doublets.

With $6$ doublets, the custodial block $C_6$ has six two-fold degenerate eigenvalues appearing in three pairs
\begin{align}
\label{eq:evals6}
(-\alpha_1,\alpha_1),\,(-\alpha_2,\alpha_2),\,(-\alpha_3,\alpha_3).
\end{align}
The corresponding 12 eigenvectors are contained in the $\So(6)$ subalgebra but cannot span it since it has dimension 15. The absence of an eigenvalue pattern for the three remaining eigenvectors which are needed to span $\So(6)$ means one would need to check if any of the $\small \binom{35-12}{3}=1771$ sets of three eigenvectors can complete the 12 remarkable eigenvectors into a basis of $\So(6)$. Moreover, because of the eigenvalue pattern~(\ref{eq:evals6}), $\So(3)$ subalgebras cannot coincide with any of the degenerate subspaces. Hence, for $N=6$, a test based on verifying F-products would be impractical because the F-products would depend on the choice of basis for the degenerate subspaces.   

Beyond $N=6$ we do not observe any eigenvalue pattern which significantly complicates the characterization of CS. However the custodial block $C_N$, being traceless for all $N$, has $k-1=\frac{(N+1)(N-2)}{2}$ independent eigenvalues which are functions of its $\small 
\binom{N}{4}$ parameters $\lambda_{abcd}$. Beyond the scope of this work lies an interesting but possibly difficult question: can any set of $k$ real numbers $\{\alpha_i \in \mathbb{R} | i=1,\ldots,k,\: \sum \alpha_i = 0\}$ be the set of roots of the characteristic polynomial of $C_N$ for some set of parameters $\{\lambda_{abcd}\}$? If this is true then the problem will simplify significantly, although the difficulty of identifying which bases of $\So(N)$ correspond to CS will remain.

\section{Summary}
\label{sec:Summary}

We have found a characterization of CS for scalar potentials with any number of doublets based on geometrical and representation-theoretical relations among the adjoint quantities $L,M$ and $\Lambda$ which characterize a potential in its bilinear form. To do so we considered the canonical form of the NHDM potential with CS and extracted an eigenvalue pattern in $\Lambda$ which naturally must be present in any Higgs basis. We then showed that CS is present when the corresponding eigenvectors coincide with particular bases of the defining representation of $\So(N)$, characterized by specific F-product relations. The task of distinguishing representations was achieved by means of embedding indices, which become apparent in F-product relations of normalized eigenvectors.

For $N\leq 5$, the presence or absence of the CS eigenvalue pattern is straightforward to identify, and we provide practical algorithms for establishing the presence or absence of CS for any numerical instance of a potential, and also for generic potentials with indeterminate coefficients, at least in the case the eigenvectors of $\Lambda$ are constant. In special cases where $\Lambda$ has highly degenerate eigenvalues, one runs into the problem of isolating Lie algebras inside of arbitrary vector spaces for which we provided a solving method. 

With more than five doublets, the CS eigenvalue pattern essentially fades away and a practical implementation of our characterization was not found. 

\section*{Acknowledgements}
RP is grateful to Igor P.~Ivanov, Celso C.~Nishi and Andreas Trautner for stimulating
discussions and helpful comments which enhanced his understanding of covariants-based
methods. MS is appreciative of Igor P.~Ivanov for recommending the covariant approach of reference
\cite{deMedeirosVarzielas:2019rrp} for recognizing CS.

\appendix
\section{Some mathematical results}
\label{sec:SomeMathematicalResults}
 A basis for a representation of a Lie algebra $\mathfrak{g}$ may be written as $\{B^i\}_{i=1}^b$, where $b$ is the number of matrices in the basis, which may be less than the dimension of $\mathfrak{g}$ if the representation is not faithful.
 In the following Lemma, we will abbreviate such a basis by $\{B^i\}$ and hence, for simplicity, suppress the range of the index $i$.
\begin{lemma}[A generalized Schur's Lemma]
\label{lemma:genSchursLemma}
Let $\{B^i\}=\{\text{diag}\,(B_1^i,\ldots,B_k^i)\}$ be a basis for a complex representation of a Lie algebra $\mathfrak{g}$ written
in block diagonal form, where each set of $n_j\times n_j$-dimensional matrices $\{B_j^i\}$ is the basis of an irreducible representation of $\mathfrak{g}$. Moreover, assume that each irreducible representation $\{B_j^i\}$ only occur once. Then a matrix $M$ which commutes with all matrices in $\{B^i\}$ will be of the form 
\begin{align}\label{E:Mdiag}
	M=\text{diag}\,(\lambda_1 I_{n_1\times n_1},\ldots,\lambda_kI_{n_k\times n_k}) 
\end{align}
for complex numbers $\lambda_j$.
\end{lemma}
\begin{proof}
  Write $M$ in block form with blocks $M_{mn}$, where the $k$ diagonal blocks have the same dimensions as
 the diagonal blocks of $\{B^i\}$. Then the ordinary Schur's Lemma gives us that each diagonal block $M_{mm}$ of $M$ must be a multiple of identity, since $M$ commutes with
 $B^i_m$ for all $i$. 

Moreover, the off-diagonal block elements (not necessarily square) of $M$ have to be zero. 
Indeed, suppose $MB^i=B^iM$ for all $i$, and consider an off-diagonal block which, consequently, satisfies 
\begin{align}\label{E:offDiagBlocks}
	B_m^i M_{mn} = M_{mn} B_n^i,
\end{align}
for all $i$, with no sum over $m$ or $n$, and with $m\ne n$. We will show by contradiction that
the matrices $M_{mn}=0$, i.e.~the off-diagonal blocks of $M$ are zero. 

Assume that $M_{mn}\ne 0$. We then claim that $M_{mn}$ has a non-trivial nullspace (i.e.~the nullspace is neither zero nor the whole space $M_{mn}$ is acting on).
Indeed, if the matrix $M_{mn}$ is square then it cannot be invertible, for then~\eqref{E:offDiagBlocks} would infer that the 
representation $\{B_m^i\}$ is equivalent to the representation $\{B_n^i\}$, contrary to the premise of the Lemma.
Since the matrix $M_{mn}$ is not invertible but non-zero, it has a non-trivial nullspace.  
On the other hand, if $M_{mn}$ is not square, the non-zero $M_{mn}$ will always have a non-trivial nullspace either by multiplying vectors from the left (cokernel) or from the right (kernel). 
Now let $W$ be the nullspace of $M_{mn}$, and assume the number of columns of $M_{mn}$ is greater or equal to the number of rows, i.e.~it has a non-trivial kernel.
Then~\eqref{E:offDiagBlocks} gives 
\begin{align}\label{E:offDiagBlocks2}
	0= M_{mn} B_n^i W,
\end{align}
  but since $\{B_n^i\}$ is irreducible, we can find an index $i$
such that $W'\equiv B_n^i W \nsubseteq W $, otherwise $W$ would be an invariant subspace. But then~\eqref{E:offDiagBlocks2} yields
$M_{mn}W'=0$ which contradicts that $W$ was the nullspace of $M_{mn}$. Hence $M_{mn}=0$.

In case the number of columns of $M_{mn}$ is less than the number of rows,
the same argument as above can be applied on the transpose of 
eq.~\eqref{E:offDiagBlocks}: $M_{mn}^T$ will then have a nullspace by multiplying from the right, and this will lead to the same contradiction as before, since $\{B_m^i\}$ generates an irreducible representation
if and only if $\{(B_m^i)^T\}$ generates an irreducible representation. The latter follows from that a representation is irreducible if and only if the dual representation is irreducible.
Hence $M$ is of the block diagonal form~\eqref{E:Mdiag}.
\end{proof}

\begin{prop}
\label{prop:faithful}
Let $\Lambda$ be a real symmetric matrix. If a subset $\{v_a\}$ of eigenvectors of $\Lambda$ provide a representation of the basis elements $\{b_a\}$ of a Lie algebra $\mathfrak g \subseteq \Su(N)$ as $\Pi(b_a) \equiv v_{ai} \lambda_i$, then $\Pi$ is a faithful representation of $\mathfrak g$. 
\end{prop}
\begin{proof}
Suppose the representation $\Pi$ is not faithful, then there exists $h \in \mathfrak g$, $h=h_a b_a \neq 0$, such that $\Pi(h) = h_a v_{ai}\lambda_i = 0$. That is, $\{v_a\}$ is not a linearly independent set. But then the eigenvectors of $\Lambda$ do not span $\mathbb{R}^{N^2-1}$ and $\Lambda$ isn't diagonalizable, contradicting the assumption that $\Lambda$ is a real symmetric matrix. 
\end{proof}

\begin{prop}
\label{prop:unitary-equiv}
Let $\Pi$ and $\Pi'$ be two Hermitian (or two anti-Hermitian), complex, $N$-dimensional representations of the same Lie algebra $\mathfrak{g}$. Furthermore, assume the representations are equivalent, i.e.~there exists a matrix $S$ such that $\Pi(X) = S\Pi'(X)S^{-1}$ for all $X\in \mathfrak{g}$, and let each irreducible component of $\Pi$ only occur one time.
Then $S$ can be chosen to be special unitary.
\end{prop}
\begin{proof}
The case where the two representations $\Pi$ and $\Pi'$ are irreducible,
was proven in~\cite{Plantey:2024yfm}, although this case will be a special case of the argument below.

If the representation $\Pi$ is reducible, we may perform a basis shift on the vector space $V=\mathbb{C}^N$ the representation is acting on, such that the matrices $\Pi(X)$ are block diagonal for all 
$X\in \mathfrak{g}$.  
By the Hermiticity (or anti-Hermiticity) of the representations, 
\begin{align}\label{E:eqvHerm}
\Pi'(X)=S^{-1}\Pi(X) S= S^\dag \Pi(X) (S^{-1})^\dag ,	
\end{align}
for all $X\in \mathfrak{g}$.
By multiplying~\eqref{E:eqvHerm} by $S$ from the left, and by $S^\dag$
from the right, we see the matrix $S S^\dagger$  
commute with the block diagonal matrices $\Pi(X)$.
By Lemma~\ref{lemma:genSchursLemma} (a generalized Schur's Lemma), 
the matrix $S S^\dagger$ then must be diagonal, where the diagonal elements of $S S^\dagger$ are numbers $\lambda_i>0$ (positive since $S S^\dagger$ is positive-definite), and where $\lambda_i$ has the same value for all $i$ corresponding to the same irreducible component (i.e.~each block) of the matrices $\Pi(X)$. 
 By dividing
each row of $S$, indexed by $i$, by $\sqrt{\lambda_i}$, we then obtain a  matrix $U$ which is unitary, since $U U^\dagger=I$. Then
$\Pi(X)=U \Pi'(X) U^\dagger$, since $(\Pi(X))_{ij}=S_{im}(\Pi'(X))_{mn}S^{-1}_{nj}=(S_{im}/\sqrt{\lambda_i})(\Pi'(X))_{mn}(S^{-1}_{nj}\cdot \sqrt{\lambda_j})=U_{im}(\Pi'(X))_{mn}U^\dagger_{nj}$, where the second equality applies that $(\Pi(X))_{ij}=0$ when $i$ and $j$ corresponds to different blocks, since $(\Pi(X))_{ij}$ was block diagonal. Finally, we can write $U=e^{i\theta}U'$ where $U'$ is special unitary, and then $U'$ is the matrix sought in the Proposition.
\end{proof}

An important special case of Proposition~\ref{prop:unitary-equiv} is then
\begin{prop}
\label{prop:unitary-equivRedIrred}
Two equivalent representations of $\So(N)$ contained in $\Su(N)$, containing only one copy of each irreducible component,
may always be related by a similarity transformation given by a special unitary matrix $U$.
\end{prop}

Before showing the next Proposition, we recall that the elements of a
adjoint vector $u_a\in \mathbb{R}^{N^2-1}$ are written $u_{ai}\equiv (u_a)_i$. 
\begin{prop}
\label{prop:rotation-invariance}
If an orthonormal set of vectors $\{t_a\}_{a=1}^3$ satisfies
\begin{align}
	\alpha F^{(t_a,t_b)} = \epsilon_{abc} t_c,
\end{align}
 for some number $\alpha$, then so does the rotated set of vectors $\{t_a' = R_{ab}t_b\}_{a=1}^3$ with $R\in \mathsf{SO}(3)$.
\end{prop}
\begin{proof}
Consider $\alpha F^{(t_a',t_b')}_k = \alpha f_{ijk}  t_{ai}' t_{bj}' 
= \alpha R_{ad}R_{be}f_{ijk}t_{di} t_{ej}=\alpha R_{ad}R_{be} F^{(t_d,t_e)}_k$ by definition. Now using a Levi-Civita symbol identity $R_{ad}R_{be}R_{cg}\epsilon_{deg} = \text{det}(R)\epsilon_{abc}$, which infers $R_{ad}R_{be} \epsilon_{deg}=\text{det}(R)\epsilon_{abc} R_{cg}$, and the assumption that $R\in \mathsf{SO}(3)$ we get
\begin{align}
\label{eq:Fprod3-rotated}
\alpha F^{(t_a',t_b')} 
&= R_{ad}R_{be} \epsilon_{deg} t_g \nonumber \\
&=  \epsilon_{abc} R_{cg}  t_g \nonumber \\
&=  \epsilon_{abc} t_c' 
\end{align}
\end{proof}
In the case of an improper rotation $R\in \mathsf{O}(3)$, an emerging minus sign in~(\ref{eq:Fprod3-rotated}) from $\text{det}(R)$ can be absorbed into the definition of the vectors $\{t_a'\}$.

The following Proposition shows that the F-product relations characterizing the custodial block $C_N$ for $N=5$ are given by setting all $\lambda_{abcd}$ of the custodial invariants 
$I^{(4)}_{abcd}$ to zero, except for one. And then we easily can decide if a 5HDM matrix $\Lambda$ is custodial-symmetric.
\begin{prop}
\label{prop:5HDM-rotations}
Let $V$ be a manifestly custodial-symmetric 5HDM potential. Then all but one custodial invariants $I^{(4)}_{abcd}$ may be eliminated through a series of orthogonal Higgs basis transformations, while $V$ is preserved in a manifestly custodial-symmetric form.
\end{prop}
\begin{proof}
 First, the part of $V$ corresponding to the custodial block may be written
\begin{align}
	V_C= \lambda_{1234} I^{(4)}_{1234}+\lambda_{1235} I^{(4)}_{1235}+\lambda_{1245} I^{(4)}_{1245}+\lambda_{1345} I^{(4)}_{1345}+\lambda_{2345} I^{(4)}_{2345}.
\end{align}
   Note that the invariant 
	\begin{align}
 I^{(4)}_{abcd}\equiv  I^{(4)}(\Phi_a,\Phi_b,\Phi_c,\Phi_d) &= \text{Im}(\Phi_a^\dag \Phi_b)\text{Im}(\Phi_c^\dag \Phi_d) 
+\text{Im}(\Phi_a^\dag \Phi_d)\text{Im}(\Phi_b^\dag \Phi_c) \nonumber \\ &+\text{Im}(\Phi_a^\dag \Phi_c)\text{Im}(\Phi_d^\dag \Phi_b),
 \end{align}
	is $\mathbb{R}$-linear in all its variables, in the sense 
\begin{align}	
	I^{(4)}(r_1x_1+r_2x_2,y,z,w) = r_1 I^{(4)}(x_1,y,z,w)+r_2 I^{(4)}(x_2,y,z,w),
\end{align}
for $r_1, r_2 \in \mathbb{R}$, and similarly for the other variables.
Without loss of generality, we will now show how to eliminate all custodial invariants $I^{(4)}$ but $I^{(4)}_{1234}$.
Consider the orthogonal basis change (a mixing of doublet 1 and 2)
\begin{align}\label{E:OrthTrafo5HDM}
\begin{pmatrix}
    \Phi_1 \\
    \Phi_2 \\
\end{pmatrix}   
\to 
\begin{pmatrix}
  \cos \alpha && - \sin \alpha \\
    \sin \alpha && \cos \alpha \\
\end{pmatrix}  
\begin{pmatrix}
    \Phi_1 \\
    \Phi_2 \\
\end{pmatrix},
\end{align}
other fields left invariant. Orthogonal basis changes act block diagonally on $\Lambda_C$,
     and do not mix the custodial block $C_5$ with $A_5$, since bilinears associated with imaginary Gell-Mann matrices are mapped to other bilinears associated with imaginary Gell-Mann matrices, while bilinears associated with real Gell-Mann matrices remain real.
		Then, under the transformation~\eqref{E:OrthTrafo5HDM},
		\begin{align}
			\lambda_{2345} I^{(4)}_{2345} + \lambda_{1345} I^{(4)}_{1345} &\to
			(\lambda_{2345} \cos \alpha - \lambda_{1345} \sin \alpha)I^{(4)}_{2345} \nonumber \\
			&+ (\lambda_{1345} \cos \alpha+\lambda_{2345}\sin \alpha)I^{(4)}_{1345}.
		\end{align}
		Hence we can eliminate one of these custodial invariants, e.g.~$I^{(4)}_{2345}$, by setting
		\begin{align}
		\alpha = \arctan(\frac{\lambda_{2345}}{\lambda_{1345}}).
		\end{align}
		The other terms associated with $C_5$,
	\begin{align}\label{E:restSum}
   		\lambda_{1234} I^{(4)}_{1234} + 	\lambda_{1235} I^{(4)}_{1235} 
			+ \lambda_{1245} I^{(4)}_{1245},
	\end{align}
		are mapped to terms of the \emph{same type} under~\eqref{E:OrthTrafo5HDM}:
		For instance will, when doublets 1 and 2 are mixed by~\eqref{E:OrthTrafo5HDM}, 
		\begin{align}
			\lambda_{1234} I^{(4)}_{1234} \to \lambda_{1234} I^{(4)}_{1+2,1+2,3,4}
		\end{align}
		where index $1+2$ means we have some $\mathbb{R}$-linear combination of 
		$\Phi_1$ and $\Phi_2$ as the corresponding variable of $I^{(4)}(x,y,z,w)$. 
		Then, since $\lambda_{abcd}$ is $\mathbb{R}$-linear in all variables,
		\begin{align}
			\lambda_{1234} I^{(4)}_{1234} \to \lambda_{1234}' I^{(4)}_{1234},
		\end{align}
		for some real number $\lambda_{1234}'$. Here we have used that $I^{(4)}_{abcd}$ is antisymmetric in all indices, which e.g.~infers $I^{(4)}_{1134}=0$. Moreover, each of the terms in the sum 
		\eqref{E:restSum} will be mapped to new terms of the exactly same type under~\eqref{E:OrthTrafo5HDM}, so~\eqref{E:restSum} is preserved in the same form. Hence, we have eliminated $I^{(4)}_{2345}$ from $V$. 
		We may then proceed in the same manner with the surviving custodial invariants of $V$,
		where $V_C$ in the new basis may be written
		\begin{align}\label{E:surviving}
		V_C= \lambda_{1234}' I^{(4)}_{1234}+\lambda_{1235}' I^{(4)}_{1235}+\lambda_{1245}' I^{(4)}_{1245}+\lambda_{1345}' I^{(4)}_{1345}.
		\end{align}
  By letting  
	\begin{align}\label{E:OrthTrafo5HDM_2}
\begin{pmatrix}
    \Phi_2 \\
    \Phi_3 \\
\end{pmatrix}   
\to 
\begin{pmatrix}
  \cos \beta && - \sin \beta \\
    \sin \beta && \cos \beta \\
\end{pmatrix}  
\begin{pmatrix}
    \Phi_2 \\
    \Phi_3 \\
\end{pmatrix},
\end{align}
	the two last terms of~\eqref{E:surviving} are rotated into each other, while the
	other terms are mapped to terms of the exactly same type (i.e.~corresponding to the same $I^{(4)}_{abcd}$). By adjusting the angle $\beta$ to an appropriate
	value, we may eliminate the last term of~\eqref{E:surviving}. We may continue in the same way until only $\lambda_{1234}'' I^{(4)}_{1234}$ is left. 
		\end{proof}
	The value of the surviving parameter in Proposition 
	\ref{prop:5HDM-rotations} will be given by
	\begin{align}
		\lambda_{1234}''=\sqrt{\sum_{a<b<c<d} \lambda_{abcd}^2},
	\end{align}
	since orthogonal basis transformations conserve the eigenvalues of $C_N$.
The procedure of Proposition 
	\ref{prop:5HDM-rotations} only works for $N=5$, since given any distinct pair of indices when N=5, there is always only two $I^{(4)}_{abcd}$ with exactly one of the numbers among their indices. Hence these two invariants will be rotated into each other while the others are left
	in the same form under an $SO(2)$ basis shift. Furthermore, Proposition~\ref{prop:5HDM-rotations} infers that all custodial blocks $C_5$ with the same eigenvalues are equivalent, since they are all equivalent to this simple instance with only one non-zero $\lambda_{abcd}$.

\section{Handling large degeneracies}
\label{sec:EigenvalueDegeneraciesBeyondTheCharacteristic}
In this Appendix, we will consider eigenvalue degeneracies beyond
the degeneracies which are characteristic of the CS. In cases where such degeneracies exist, one runs into the problem of searching for Lie algebras within a generic vector space i.e. identifying subspaces which are also Lie algebras. While a solution based on Lie algebraic methods, for instance involving root systems, would be most satisfying, the authors are not aware of any theory on this subject, when the ambient vector space $V$ itself is not a Lie algebra. Therefore we propose below a solution based on solving systems of quadratic polynomial equations. Our method relies on transforming the problem into the minimization of a quartic polynomial which may have up to 90 variables. Even with so many variables, the minimization is straightforward with e.g.~\texttt{Scipy}'s~\cite{2020SciPy} optimization module and we manage with a naive implementation to solve the relevant equations even for the most extreme degeneracy patterns in the 5HDM in a couple of minutes on an ordinary desktop computer. It is likely that the computation time can be reduced with more sophisticated optimization code. 

\subsection*{\boldmath $N=3$}
In the 3HDM, when there are more $LM$-orthogonal nullvectors than the three 
that are characteristic for the CS, i.e.~$l\equiv\text{dim}(W_0^{LM})>3$,
then three linear combinations of the basis vectors of $W_0^{LM}$ might generate the defining representation of $\So(3)$, which is necessary and sufficient for CS. To isolate these linear combinations, if they exist, we begin by considering three arbitrary vectors of $W_0^{LM}$
\begin{equation}
\label{eq:general-linear}
v_i = c_{ij} u_j \, , \quad i\in \{1,2,3\}.
\end{equation}
where $\{u_j\}_{j=1}^l$ is any orthonormal basis of $W_0^{LM}$ and $c_{ij}$ are coefficients to be determined. If the vectors~\eqref{eq:general-linear} are to form an orthonormal basis for the defining representation of $\So(3)$, then they must satisfy the following equations
\begin{align}
\label{eq:app3HDMeqs}
g_{ab}(c) &\equiv v_a \cdot v_b -\delta_{ab} = 0 \, , \quad b \leq a \leq 3 \nonumber \\
h_{ab}(c) &\equiv 2F^{(v_a,v_b)} - \epsilon_{abd} v_d = 0 \, , \quad b < a \leq 3.  
\end{align}
This system of 30 equations is to be solved for the 3$l$ coefficients $c_{ij}$, which can be difficult using a direct solving approach or even Gröbner bases~\cite{Cox2015}. We find that the most robust method for finding numerical solutions, if they exist, is to transform the problem into an optimization problem by defining a cost function 
\begin{equation}
J \equiv \sum_{b\leq a \leq 3} g_{ab}^2 + \sum_{b < a \leq 3} h_{ab} \cdot h_{ab} 
\end{equation}
which is to be minimized with respect to the coefficients $c_{ij}$. Solutions to the equations~\eqref{eq:app3HDMeqs} then correspond to minima of the cost function with $J=0$. Conversely, if $J>0$ at its global minimum, then there are no solutions. Such optimization problems are very well studied, and there exist many algorithms to tackle them, which are implemented in readily available computing packages. Large degeneracies in the 4HDM and 5HDM may be treated in the same way, with the appropriate equations.

\subsection*{\boldmath $N=4$}
In the case of the 4HDM with extra degeneracies such that $l^+ \equiv \text{dim}(W^{LM}_{+\alpha}) > 3$ or $l^- \equiv \text{dim}(W^{LM}_{-\alpha}) > 3$, one must, as described in Section~\ref{sec:N4}, look for six orthonormal vectors, three in $W^{LM}_{+\alpha}$ and three in $W^{LM}_{-\alpha}$, generating the defining representation of $\So(4)$. As before, we parametrize these vectors as
\begin{align}
v^\pm_i = c^\pm_{ij} u^\pm_j \, , \quad i\in \{1,2,3\} 
\end{align}
where $\{u^\pm_j\}_{j=1}^{l^\pm}$ are bases for $W^{LM}_{\pm \alpha}$ and $c^\pm_{ij}$ are coefficients to be determined. Now one must find out whether or not the equations
\begin{align}
g^{(\pm)}_{ab}(c^\pm) &\equiv v^\pm_a \cdot v^\pm_b -\delta_{ab} = 0 \, , \quad b \leq a \leq 3 \nonumber \\
h^{(++)}_{ab}(c^+) &\equiv \sqrt{2}F^{(v^+_a,v^+_b)} - \epsilon_{abc} v^+_c = 0 \, , \quad b < a \leq 3 \nonumber \\ 
h^{(--)}_{ab}(c^-) &\equiv \sqrt{2}F^{(v^-_a,v^-_b)} - \epsilon_{abc} v^-_c = 0 \, , \quad b < a \leq 3 \nonumber \\
h^{(+-)}_{ab}(c^\pm) &\equiv \sqrt{2}F^{(v^+_a,v^-_b)} = 0 \, , \quad a,b \leq 3.
\end{align}
have any solutions. Following the same optimization strategy as in the 3HDM to solve what is now a system of 237 quadratic equations with $3(l^++l^-)$ unknowns, the cost function to minimize is
\begin{equation}
J \equiv \sum_{b\leq a \leq 3} \Big(g^{(+)2}_{ab}+g^{(-)2}_{ab}\Big) + \sum_{b < a \leq 3} \Big(h^{(++)}_{ab} \cdot h^{(++)}_{ab} + h^{(--)}_{ab} \cdot h^{(--)}_{ab} \Big) + \sum_{a,b \leq 3} h^{(+-)}_{ab} \cdot h^{(+-)}_{ab}
\end{equation}

\subsection*{\boldmath $N=5$}

For the 5HDM, we may have $l^+ \equiv \text{dim}(W_{+\alpha}^{LM}) > 3$, $l^- \equiv \text{dim}(W_{-\alpha}^{LM}) > 3$ or $l_0 \equiv \text{dim}(W_{0}^{LM}) > 4$, in which case isolating the defining representation of $\So(5)$ is not as straightforward as without excessive degeneracies, cf.~Section~\ref{sec:N5}. Such extra degeneracies are handled similarly as with $N=3$ and $N=4$ doublets, by first writing down a general parametrization of three vectors in $W_{+\alpha}^{LM}$, three vectors of $W_{-\alpha}^{LM}$ and four vectors of $W_{0}^{LM}$
\begin{align}
v^\pm_i &= c^\pm_{ij} u^\pm_j \, , \quad i\in \{1,2,3\} \nonumber \\
v^0_i &= c^0_{ij} u^0_j \, , \quad i\in \{1,\ldots,4\},
\end{align}
where $\{u^\pm_j\}_{j=1}^{l^\pm}$ and  $\{u^0_j\}_{j=1}^{l^0}$ are bases for $W_{\pm\alpha}^{LM}$ and $W_0^{LM}$, and then checking if the coefficients $c^\pm_{ij}$, $c^0_{ij}$ can take values such that the ten vectors above form an orthonormal basis for the defining representation of $\So(5)$. This amounts to solving the equations
\begin{align}
\label{eq:so5eqs}
g^{(\pm)}_{ab}(c^\pm) &\equiv v^\pm_a \cdot v^\pm_b -\delta_{ab} = 0 \, , \quad b \leq a \leq 3 \nonumber \\
g^{(0)}_{ab}(c^0) &\equiv v^0_a \cdot v^0_b -\delta_{ab} = 0 \, , \quad b \leq a \leq 4 \nonumber \\
h_{ab}(c^\pm,c^0) &\equiv \sqrt{2} F^{(v_a,v_b)} - f_{abc}v_c = 0 \, , \quad b < a \leq 10
\end{align}
where in the last equation we have let $\{v_a\}_{a=1}^{10} \equiv \{v_1^+,v_2^+,v_3^+,v_1^-,v_2^-,v_3^-,v_1^0,\ldots,v_4^0\}$ for conciseness and $f_{abc}$ are structure constants of $\So(5)$ such that
\begin{equation}
\label{eq:so4inso5}
f_{abc} = \epsilon_{abc} \, , \quad 1 \leq a,b,c \leq 3 \quad \text{and} \quad 4 \leq a,b,c \leq 6. 
\end{equation}
Any such structure constants will do since the F-products involving the nullvectors $\{v_a^0\}_{a=1}^4$ are unconstrained by CS. One may, for example, choose the structure constants in the orthonormal $\So(5)$ basis given by
\begin{equation}
\bigg\{\frac{\lambda_1-\lambda_8}{\sqrt{2}},\,\frac{\lambda_2+\lambda_6}{\sqrt{2}},\,\frac{\lambda_3-\lambda_5}{\sqrt{2}},\,\frac{\lambda_1+\lambda_8}{\sqrt{2}},\,\frac{-\lambda_2+\lambda_6}{\sqrt{2}},\,\frac{\lambda_3+\lambda_5}{\sqrt{2}},\,\lambda_4,\lambda_7,\lambda_9,\lambda_{10}\bigg\}
\end{equation}
where $\lambda_i$ are the antisymmetric Gell-Mann matrices in 5 dimensions, as given in Section \ref{sec:cs-pot} and in \cite{Solberg:2018aav}. This is a convenient choice since the structure constants in this basis are sparse and satisfy~\eqref{eq:so4inso5}.

Solving the 1102 equations in~\eqref{eq:so5eqs} for the $3(l^++l^-)+4l^0$ coefficients $c^\pm_{ij},c^0_{ij}$ is then done by minimizing the cost function
\begin{equation}
J \equiv \sum_{b\leq a \leq 3} \Big(g^{(+)2}_{ab}+g^{(-)2}_{ab}\Big) + \sum_{b\leq a \leq 4} g^{(0)2}_{ab} + \sum_{b<a\leq 10} h_{ab}\cdot h_{ab}.
\end{equation}
For reference, we solved the most difficult case $l^+=l^-=3$ and $l^0=18$, where $J$ has 90 variables, in a couple of minutes on an ordinary desktop computer, for random and completely generic numerical potentials. It is also worth mentioning that, for a fixed number of variables, the number of equations, although rather impressive, does not significantly increase the difficulty of the optimization problem since the cost function is always a quartic polynomial.

\bibliographystyle{JHEP}

\bibliography{ref}

\end{document}